\documentclass[11pt]{article}

\usepackage{amsmath,amsfonts,amsthm,amssymb}
\usepackage[margin=1in]{geometry}
\usepackage{float}
\usepackage{graphicx}
\usepackage{cite}

\newtheorem{theorem}{Theorem}

\newtheorem{claim}{Claim}

\newcommand{\sumd}{\displaystyle\sum}

\title{Online Lower Bounds via Duality}

\author{
Yossi Azar \thanks{School of Computer Science, Tel-Aviv University.
E-mails: \texttt{azar@tau.ac.il}, \texttt{ilanrcohen@gmail.com}, \texttt{alan.roytman@cs.tau.ac.il}.\newline
}
\and
Ilan Reuven Cohen\footnotemark[1]
\and
Alan Roytman\footnotemark[1]
}

\date{}

\begin{document}

\maketitle
\thispagestyle{empty}
\begin{abstract}
In this paper, we exploit linear programming duality in the online setting (i.e., where input arrives on the fly)
from the unique perspective of designing lower bounds on the competitive ratio.  In particular, we provide a general
technique for obtaining online deterministic and randomized lower bounds (i.e., hardness results) on the competitive
ratio for a wide variety of problems.  We show the usefulness of our approach by providing new, tight lower bounds for three diverse
online problems.  The three problems we show tight lower bounds for are the Vector Bin Packing problem, Ad-auctions (and various
online matching problems),
and the Capital Investment problem.
Our methods are sufficiently general that they can also be used to reconstruct existing lower bounds.

Our techniques are in stark contrast to previous works, which exploit linear programming duality to obtain positive results,
often via the useful primal-dual scheme.
We design a general recipe with the opposite aim
of obtaining negative results via duality.
The general idea behind our approach is to construct a primal linear program based on a collection of input sequences,
where the objective function corresponds to optimizing the competitive ratio.
We then obtain the corresponding
dual linear program and provide a feasible solution, where the objective function yields a lower bound
on the competitive ratio.
Online lower bounds are often achieved by adapting the input sequence according to an
online algorithm's behavior and doing an appropriate ad hoc case analysis.  Using our unifying techniques, we simultaneously
combine these cases into one linear program and achieve online lower bounds via a more robust analysis.
We are confident that our framework can be successfully applied
to produce many more lower bounds for a wide array of online problems.
\end{abstract}

\section{Introduction}
In this work, we develop a framework that illustrates how to use linear programming duality to obtain online lower bounds.
That is, we provide a general technique to obtain hardness results on the competitive ratio for a wide variety of problems.
In contrast, previous works have mainly applied linear programming duality in the context of designing
algorithms that obtain positive results on the competitive ratio.  We demonstrate our approach by exhibiting new, tight online
lower bounds for several different problems.  We also apply our techniques to obtain improved lower bounds for various matching problems.
Our approach can also be used to reconstruct existing lower bounds for a range of problems, such as
online buffer management problems~\cite{BCCFJLST04}, two-sided online bipartite matching~\cite{WW15},
various load balancing problems, and more.  We are confident that the techniques we provide are capable of
yielding many additional lower bounds for a rich set of online problems.
Recall that the benchmark we use to measure the performance of an online algorithm is the \emph{competitive ratio}.  In particular, we compare
the objective value of an online algorithm $ALG$ relative to the objective value of an optimal solution $OPT$ that is omniscient and knows
the entire input sequence in advance.  More formally, for any input sequence $\sigma$, let $ALG(\sigma)$ denote
the value of $ALG$ on input sequence $\sigma$ and $OPT(\sigma)$ denote the value of an optimal solution on the same input sequence.
We say that $ALG$ is $c$-competitive if $ALG(\sigma) \leq c \cdot OPT(\sigma) + a$ for any input sequence $\sigma$ (where we allow some additive constant $a$).
For randomized algorithms, the definition is the same, except that we consider $\mathbb{E}[ALG(\sigma)]$ instead of $ALG(\sigma)$.

Previous works have used linear programming duality from the perspective of designing algorithms, thus obtaining positive results
on the competitive ratio.  In particular, the foundational primal-dual framework has been used extensively
to great success in the design and analysis of algorithms
for many different settings, and builds on the rich theory developed for the primal-dual framework in the offline setting.
The method has been applied to design exact algorithms (e.g., computing flows
in networks) along with developing approximation algorithms, beginning with the fundamental work of
Goemans and Williamson~\cite{GW95a}.  Since then, the primal-dual framework has also played an important role in designing
algorithms in the online setting, providing strong guarantees on the competitive ratio for a wide variety of problems.
This includes well-known online problems such as metrical task systems~\cite{BBN10}, the $k$-server problem~\cite{BBN10,BBN10a}, weighted
paging~\cite{BBN10b}, set cover~\cite{AAABN03}, and load balancing~\cite{ABFP13,MRT13} to name a few.

In this paper, we demonstrate our techniques on three online problems and show how to obtain tight lower bounds for each of them.
Our unifying framework can be summarized via the following recipe.  Given an online problem,
we first construct some parameterized collection of input sequences for the problem and then encode various constraints that any feasible
algorithm (with some competitive ratio guarantee) must obey via a corresponding parameterized primal linear program.  We set the objective function
of each primal linear program in such a way that optimizing it is equivalent to optimizing the competitive ratio.
Considering the parameterized dual linear program and obtaining \emph{any} feasible solution
to it yields a valid lower bound on the competitive ratio of any algorithm.
Our techniques allow us to sidestep ad hoc case analysis that is typically done to obtain online lower bounds,
where the future input sequence is adaptively determined based on the algorithm's behavior in the past.
In particular, we simultaneously encode and combine these cases via linear programming in such a way that producing strong online lower
bounds essentially boils down to finding good solutions to linear programs.
We now describe the three problems to which we apply our framework to obtain new lower bounds.

{\bf Vector Bin Packing:} The first problem we consider
is the classic Bin Packing problem in an online, multidimensional setting, and we refer to this
problem as the Vector Bin Packing problem.  In this problem, we are given a set of vectors $\{v_1,\ldots,v_n\}$
where $v_i = (v_{i}(1),\ldots,v_{i}(d)) \in [0,1]^d$
for all $i \in [n]$ (here, $[n] = \{1,\ldots,n\}$).  We must partition the vectors into feasible sets
$B_1,\ldots, B_m$ such that, for each $1 \leq j \leq m$ and each coordinate $k$, we have $\sum_{i \in B_j}v_{i}(k) \leq 1$.
We refer to each set $B_j$ as a bin, each of which has capacity $1$ for each coordinate $1 \leq k \leq d$.  The objective
function is to minimize $m$, the number of bins used to pack all vectors in a feasible manner.  In the online setting,
$d$-dimensional vectors arrive in an online manner (i.e., on the fly) and must be immediately assigned to an open bin, or to a new bin, so that
the capacity constraints on each bin are satisfied along every dimension.
This problem has applications in cloud computing~\cite{PTU11}, where the use of huge data centers have become more prevalent,
and the costs of providing power and cooling servers have skyrocketed.  In fact, these costs
now exceed the costs of purchasing hardware and servers~\cite{PN08}.  Moreover, representing
jobs as vectors captures the fact that resource usage is multidimensional (e.g., CPU, memory, and I/O).
Assigning multidimensional jobs to machines also has applications for
implementing databases for shared-nothing environments~\cite{MI96},
and optimizing parallel queries in databases, since such tasks typically consume multiple resources~\cite{MI95}.

In the work of~\cite{ACFR16}, they focused
on the setting where all vectors have small values in each coordinate relative to the size of a bin.  In particular,
they gave a $(1+\epsilon)e$-competitive algorithm for arbitrarily large $d$ when all vectors have coordinate values
smaller than $O\left(\frac{\epsilon^2}{\log d}\right)$, for any $\epsilon > 0$ (here, $e$ denotes the base of the natural logarithm).
They also defined a \emph{splittable} model, where a vector $v$ can be split into arbitrarily many fractions
$v\cdot \alpha_1, v\cdot \alpha_2, \ldots, v\cdot \alpha_k$, $\sum_i \alpha_i =1$ (here, each
$v\cdot \alpha_i$ can be placed into a different bin).  In this setting, they gave an $e$-competitive algorithm.
We show how to use our techniques to produce a matching lower bound of $e$ in the splittable setting,
which implies the same lower bound for the original problem (even when vectors are small).

{\bf Ad-auctions:} The second problem we provide a tight lower bound for is the online Ad-auctions problem.  In this problem,
there are $n$ bidders which are known up front.  Each bidder $i$
has a budget of $B(i)$.  In addition, products arrive in an online manner (one product at a time).
As each product $j$ arrives, each buyer $i$ provides a bid $b_{i,j}$ for the product $j$.  The mechanism
must then allocate product $j$ to a buyer, and gains a revenue of $b_{i,j}$ (the decision of
which buyer $i$ receives item $j$ is irrevocable).  Moreover, buyers cannot be charged more than their
total budget.  The objective is to maximize the total revenue.
In the fractional version of the problem, the algorithm may sell a fraction of each item $j$ to
multiple buyers (as long as the sum of all fractions does not exceed one), and the
revenue received from each buyer is appropriately scaled according to the fraction sold.
This problem was introduced in~\cite{MSVV05} and is very important for search engine companies that wish to maximize revenue.  In practice, advertisers associate
their advertisements with search keywords.  They then place bids on the amount they pay whenever a
user clicks on the advertisement.  In particular, when a user submits a search query to the
search engine, an ad-auction is run in which advertisements are immediately allocated
to advertisers (see~\cite{R05} for further motivation).

In~\cite{BJN07}, they gave a $\left(1-\frac{1}{e}\right)$-competitive algorithm for this problem, which matches
the lower bound given in~\cite{MSVV05}.  The work of~\cite{BJN07} also studied the bounded degree setting,
where the number of buyers who are interested in each product is bounded by some parameter $d$.
Under this constraint, they obtained an improved competitive ratio of $1 - \left(1-\frac{1}{d}\right)^d$.  Note that this competitive ratio approaches $\left(1-\frac{1}{e}\right)$
from above for arbitrarily large $d$.  We provide a matching lower bound of $1 - \left(1 - \frac{1}{d}\right)^d$
for this problem.

{\bf Capital Investment:} The third problem we consider is the online Capital Investment problem (also known
as the multislope ski rental problem~\cite{LPR12}), which is a generalization of
the ski rental problem.  In this setting, we wish to produce many units of a particular
commodity at minimum cost.  Over time, orders for units of the commodity arrive in
an online manner, where each unit is characterized by its arrival time.
In addition, we have a set of machines (which are available in advance), where each machine $m_i$ is characterized by its
capital cost $c_i$ and its production cost $p_i$, and can produce arbitrarily many units
of the commodity.  At any
time, the algorithm can choose to buy any machine for a capital cost of $c_i$.  Moreover, the algorithm incurs a production
cost of $p_i$ if it uses machine $m_i$ to produce one unit of the commodity.  An algorithm must decide which machines to purchase and when to
purchase machines, in order to minimize the total cost: the sum of capital costs
plus production costs.  This problem has applications in manufacturing and making investments
in the future so as to minimize costs.  In addition, the problem has
many applications in financial settings~\cite{EFKT92} and asset allocation problems~\cite{R92}.

In~\cite{ABFFLR96}, they studied the setting in which machines arrive in an online manner and can only
be purchased after arrival.  They gave an $O(1)$-competitive algorithm for the problem assuming the case
that lower production costs means higher capital costs.
In~\cite{LPR12}, they
studied the problem where all machines are available up front and gave an $e$-competitive
algorithm for this setting.  Using our techniques, we give a matching lower
bound of $e$.

\subsection*{Contributions and Techniques}
Our main contribution is a general technique that enables us to obtain online lower bounds via
duality.  Using our methods, we obtain the following lower bounds.

\begin{enumerate}
\item {\bf Vector Bin Packing:} We give a tight lower bound on the competitive ratio for the Vector Bin Packing problem
when vectors are small (even in the splittable, randomized setting) which approaches $e$ for arbitrarily large $d$ (here, $d$ denotes
the dimension of each vector).  This improves upon the previous best lower bound of $\frac{4}{3}$
and matches the positive result obtained in~\cite{ACFR16}.

\item {\bf Ad-auctions:} We give a tight lower bound on the competitive ratio for the online
Ad-auctions problem in the bounded degree setting, where $d$ is an upper bound on the number of buyers
who are interested in each product.  We give a lower bound of $1 - \left(1-\frac{1}{d}\right)^d$ against randomized
algorithms, which improves upon the previous best result and matches the positive result given in~\cite{BJN07}.

\item {\bf Capital Investment:} We give a tight lower bound on the competitive ratio for the online
Capital Investment problem.  In particular, we give a lower bound against randomized algorithms that is
arbitrarily close to $e$, which improves over the previous best result of $\frac{e}{e-1}$ (i.e.,
the randomized ski rental lower bound~\cite{KMMO94}) and matches the positive result provided in~\cite{LPR12}.
\end{enumerate}
Our results yield more implications.
\begin{itemize}
\item {\bf Bipartite Matching: }
For the online Ad-auctions problem, our lower bound also implies the same lower bound
against randomized algorithms for the online matching problem~\cite{KVV90} in the bounded degree setting,
where the degree of each arriving node is bounded by~$d$.

\item {\bf Instance-tight: }
For the online Capital Investment problem,
our techniques are sufficiently general that we can provide a different lower bound for each
input configuration (i.e., values for machine capital costs and production costs) that is tight against each such
configuration.
\end{itemize}
A few remarks regarding our results are in order, which may be of independent interest.
\begin{itemize}
\item {\bf Fractional versus Randomized: }
Typically, deterministic fractional lower bounds also imply lower bounds for randomized algorithms since we can
view fractions as the expectation of probabilities (with the appropriate definition of a fractional algorithm).
This easily holds for the online Ad-auctions problem and the online Capital Investment problem.
For the Vector Bin Packing problem, the definition of an appropriate fractional algorithm is more
subtle (see Appendix~\ref{apx:rvbp}), but this statement still holds.

\item {\bf Additive Term: }
For the online Ad-auctions problem, an additive constant in the competitive ratio cannot reduce the
competitive ratio since we can duplicate any lower bound example many times.  For the online Capital Investment problem,
an additive constant cannot reduce the competitive ratio since we can simply scale all costs.  For the Vector Bin Packing problem,
an additive constant can reduce the competitive ratio in general.
However, our lower bound also holds if an additive constant is allowed.  This is achieved in our lower bound
since all input vectors are duplicated many times.  Hence, in all of our lower bounds, we assume that such additive constants are zero.
\end{itemize}

Our technique for providing online lower bounds can be summarized as follows.  Given an online problem,
the first step is to construct some parameterized collection of input sequences for the problem.  Then, given a supposed algorithm
that solves the problem with some competitive ratio, we express constraints that any feasible algorithm
must satisfy as a parameterized primal linear program with an appropriate objective function.
Namely, the objective should be such that optimizing it corresponds to optimizing the competitive ratio.

At this point, we can run any linear program solver on the parameterized primal linear program to obtain a computational lower bound
on the competitive ratio for the problem.  However, we can only run the linear program solver on finitely
many primal linear programs (i.e., for finitely many parameter values of our parameterized program).  Hence,
if we wish to obtain a lower bound for arbitrarily large parameter values (which is useful since our lower bounds
improve as the parameter increases), we must provide a formal, mathematical proof.  It is possible
to obtain such a proof by finding an optimal solution to the parameterized primal linear program, but finding
such a solution and proving it is optimal can be difficult.  To reconcile this issue, we consider the corresponding parameterized dual linear program.
Now, obtaining \emph{any} feasible solution to the dual linear program (which is significantly easier than
finding an optimal solution to the primal linear program) yields a valid lower bound on the competitive ratio of any algorithm.
As we obtain feasible dual solutions that are closer to the optimal dual solution, we improve our online lower bound.

\subsection*{Related Work}
There is a large body of literature concerning the offline primal-dual framework.
In the fundamental work of~\cite{GW95a}, they gave a $2$-approximation algorithm for a wide variety of problems,
including the minimum-weight perfect matching problem, the $2$-matching problem,
 and the prize-collecting traveling salesman problem.  The primal-dual framework
has also been applied in the online setting in many previous works.
In~\cite{AAABN03}, they studied the online set cover problem and presented an
$O(\log m \log n)$-competitive algorithm, where $n$ denotes the size of the ground
set and $m$ denotes the number of subsets.
In~\cite{BBN10}, they gave an $\exp(O(\sqrt{\log \log k \log n}))$-competitive algorithm
for the $k$-server problem on $n$ uniformly spaced points on a line (where $k$
denotes the number of servers).  In~\cite{BBN10a}, they gave an $(r+O(\log k))$-competitive
algorithm for the finely competitive paging problem, where $k$ denotes the cache size
and the offline algorithm pays a cost of $\frac{1}{r}$ for renting (while the online
algorithm pays a cost of $1$).  In~\cite{BBN10b}, they gave an $O(\log k)$-competitive randomized
algorithm for the weighted
paging problem, where $k$ denotes the cache size.  The work of~\cite{ABFP13} considered
a generalized framework, which is referred to as the Online Mixed Packing and Covering (OMPC) problem,
where they considered linear programs with packing constraints and covering constraints.
There are also some results in load balancing for unrelated machines with
activation costs~\cite{ABFP13,MRT13}.  Please see~\cite{BN09} for a more comprehensive
treatment of the literature on the online primal-dual framework.

{\bf Vector Bin Packing:} There is a large body of work for the Vector Bin Packing problem.
The online problem was studied in Azar et al.~\cite{ACKS13}.
They gave a lower bound of $\Omega\left(d^{\frac{1}{B} - \epsilon}\right)$ on the competitive ratio of any algorithm (for any $\epsilon > 0$),
and designed an online algorithm that is $O\left(d^{\frac{1}{B-1}}(\log d)^{\frac{B}{B-1}}\right)$-competitive
for any integer $B \geq 2$ (here, $B$ denotes the ratio between the largest coordinate and each bin's capacity).
In~\cite{ACFR16}, they studied the same problem where vectors are small relative to each bin's capacity.  For arbitrarily
large $d$, they gave a $(1+\epsilon)e$-competitive algorithm whenever each vector's coordinates are at most
$O\left(\frac{\epsilon^2}{\log d}\right)$, for any $\epsilon > 0$.
In the splittable model, they gave an $e$-competitive
algorithm for arbitrarily large $d$.
In the offline setting, the work of~\cite{CK99} gave an algorithm that achieved an
$O\left(1 + \epsilon d + \ln\left(\frac{1}{\epsilon}\right)\right)$-approximation in polynomial time
for large $d$.
For the single dimensional setting, where $d=1$ (i.e., Bin Packing), there is a vast body of literature.
There are surveys available for the offline and online settings, along with some multidimensional results
and other models~\cite{GW95,CCGMV13}.

{\bf Ad-auctions:} The online Ad-auctions problem has also received a great deal of attention in the community.
It was first introduced by~\cite{MSVV05}, where they gave a deterministic $\left(1- \frac{1}{e}\right)$-competitive algorithm for the problem,
under the assumption that each bidder's budget is large relative to their bids.  They built
on the works in online bipartite matching~\cite{KVV90} and online $b$-matching~\cite{KP00}
(which is a special case of the online Ad-auctions problem, where all bidders have the same budget $b$
and all bids are $0$ or $1$).
The work of~\cite{BJN07} gave a generalized $\left(1-\frac{1}{e}\right)$-competitive algorithm for the problem within the
primal-dual framework, which matches the bounds
given in~\cite{MSVV05}.  In addition, they gave a deterministic $\left(1-\frac{1}{e}\right)$-competitive
fractional algorithm for the online matching problem in bipartite graphs.  They
also considered multiple extensions to the online Ad-auctions framework.  These
include the ability to handle multiple slots, in which ad-auctions can be allocated
to $\ell$ slots, and buyers can submit slot dependent bids on keywords (their algorithm
maintains the same competitive ratio guarantee).  They also considered the bounded degree setting,
in which they gave a $\left(1 - \left(1 - \frac{1}{d}\right)^d\right)$-competitive algorithm (where $d$
denotes an upper bound on the total number of buyers who are interested in each product).
There are many other works, in both the offline and online settings, that have considered maximizing the
revenue of a seller in various models~\cite{AM04,BH05,BCIMS05,MS06}.

{\bf Capital Investment:} The online Capital Investment problem (sometimes referred to as the multislope ski rental problem~\cite{LPR12})
and its variants have also been studied extensively in the literature.
In~\cite{D03}, they gave a $4$-competitive deterministic algorithm and a $3.618$ lower bound, along with a $2.88$-competitive
randomized algorithm.  In~\cite{ZPX11}, they studied a similar model where machines have some duration and can expire.
They gave a $4$-competitive algorithm for this problem, along with a deterministic matching lower bound.
In~\cite{LPR12}, they studied the problem where machines
are not necessarily bought from scratch.  They provided an $e$-competitive
algorithm for this setting.  Under an additive assumption regarding machine capital costs, they gave
an improved algorithm with a competitive ratio of $\frac{e}{e-1}$.
In~\cite{LP15}, they studied the case when there are only two machines and gave matching upper and lower bounds
on the competitive ratio for deterministic and randomized algorithms.
In~\cite{ABFFLR96}, they gave an $O(1)$-competitive algorithm for the problem where machines arrive online, assuming the case
that lower production costs implies higher capital costs.  On the other hand, if both capital and production costs drop,
they gave an $O(\min\{\log C, \log \log P, \log M\})$-competitive algorithm, where $C$ is the
ratio of the highest to lowest capital costs, $P$ is the ratio of the highest to lowest production costs,
and $M$ is the number of machines that arrive online.
The online mortgage problem has also been studied~\cite{EK93}, which is similar to the
online Capital Investment problem except that future demand is known and capital costs
are fixed.

\section{Multidimensional Vector Bin Packing}\label{sec:vbp}
In this section, we study the online $d$-dimensional Vector Bin Packing problem in the splittable setting.
In this problem, we are given vectors $\{v_1,\ldots,v_n\}$ that arrive in an online manner, where $v_i = (v_{i}(1),\ldots,v_{i}(d)) \in [0,1]^d$
for all $i \in [n]$ (recall that $[n] = \{1,\ldots,n\}$).  We must assign incoming vectors into bins $B_j$ such that, for each bin $B_j$ and each coordinate $k$,
we have $\sum_{i \in B_j}v_{i}(k) \leq 1$.  The goal is
to minimize the number of bins opened to feasibly pack all vectors.  In the splittable model, a vector $v$
can be split into arbitrarily many fractions, $v\cdot \alpha_1, v\cdot \alpha_2, \ldots, v\cdot \alpha_k$,
$\sum_i \alpha_i =1$ (here, each $v\cdot \alpha_i$ can be placed into a different bin).
In the splittable model, $OPT$ is easy to compute and is given by $OPT = \max_k\lceil\sum_i{v_{i}(k)}\rceil$ (see~\cite{ACFR16}).
Moreover, the assumption that each vector's entries are in $[0,1]$ is irrelevant, as any big vector can be split
into many smaller vectors.

We give a lower bound that is arbitrarily close to $e$ for the Vector Bin Packing problem (the lower
bound approaches $e$ for arbitrarily large $d$).  The lower bound
we present in this section is against deterministic fractional algorithms, which corresponds to the splittable model.
This implies the same lower bound  for the original problem (i.e., the non-splittable setting).  We show how to adapt
our lower bound techniques to obtain a lower bound against randomized algorithms in Appendix~\ref{apx:rvbp}.
Before applying our technique, we prove the following useful claim.

\begin{claim}
	\label{cla:f}
	For any function $f(x)$ such that $f'(x)$ is a monotone non-increasing function, the following holds for all integers $i \geq j$:
	$$   \sum_{r=j}^{i-1} f'(r+1) \leq f(i) - f(j) \leq \sum_{r=j}^{i-1} f'(r) .$$
\end{claim}
\begin{proof}
	We have
	$f(i) - f(j) = \int_j^i f'(x) dx = \sum_{r=j}^{i-1} \int_r^{r+1} f'(x) dx$
	and $f'(r+1) \leq \int_r^{r+1} f'(x)dx \leq f'(r)$ since $f'(x)$ is monotone non-increasing.
\end{proof}
Our proof of the lower bound holds against any algorithm in the splittable setting that can open fractions of bins
(i.e., where the capacity constraints are appropriately reduced).
\begin{theorem}
There does not exist a deterministic fractional algorithm with a competitive ratio strictly better than $e$ for the $d$-dimensional Vector Bin Packing
problem, where $d$ is arbitrarily large.
\end{theorem}
\begin{proof} $ $\\
\indent	\textbf{Sequence definition:} 	We give a sequence $\sigma(d)$ for which the competitive ratio approaches $e$ for any algorithm
as $d$ increases.  For a fixed $d$, the sequence consists of $d$ phases. In each phase $i$, we have $A$ vectors of type $v_i$ arrive
where $v_i$ is given by: $v_1 = (1, 0, 0, \ldots , 0)$, $v_2 = (1, 2, 0, \ldots , 0)$, $v_3 = (1, 1, 3, \ldots , 0)$, $\ldots$,
$v_d = (1, 1, 1, \ldots , d)$.
More formally, 
$$v_i(k) =  \begin{cases} 
1 & \textrm{if } k < i,\\
i&  \textrm{if } k = i,\\
0 & \textrm{if } k > i.
\end{cases}
$$

Clearly, $OPT$ is $j\cdot A$ after the $j^{th}$ phase. Therefore, any $c$-competitive algorithm must open at most a total amount of $c \cdot j \cdot A$ fractional
bins (i.e., total sum of fractions) by the end of the $j^{th}$ phase.  For any deterministic algorithm, we can assume without loss of generality that an
online algorithm opens exactly a
$c \cdot j \cdot A$ fractional amount of bins.  Therefore, in each phase $j$, the algorithm opens a $c \cdot A$ fractional amount of additional bins.
Note that we only analyze the algorithm's behavior at the end of each phase.

\textbf{The primal linear program:}
\begin{itemize}
	\item $x_{i,j}$ -- the total fraction of vectors of type $v_i$ that the online algorithm assigns
to bins that are opened in phase $j$ ($i \geq j$).
	\item $c$ -- a variable representing the competitive ratio guarantee.
\end{itemize}
When writing our linear programs, we refer to each constraint by its associated dual variable.
\begin{align*}
  \text{min   } & c \\
\text{s.t.: }&  \sum_{r=j}^d v_r(k) x_{r,j} \leq c& &\text{constraint $z_{k,j}$}& &\forall k,j \in [d]& \\
 &\sum_{r=1}^i x_{i,r} = 1   && \text{constraint $y_i$}& &\forall i \in [d]& \\
 &c,x_{i,j} \geq 0 & && &\forall i\geq j : i,j \in [d], &
\end{align*}
where the constraint $z_{k,j}$ corresponds to the volume constraint of coordinate $k$ in bins opened during phase
$j$ (note that the factor of $A$ is canceled).  In particular, for any dimension $k$ and for any bins opened
in phase $j$, the total volume that can be placed on these bins is at most $c \cdot A$ since this is how much
space is available.  The constraint corresponding
to $y_i$ says that all vectors of type $v_i$ must be fully assigned to bins opened in phases $1$ through $i$. 
We omit constraints $z_{k,j}$ for $k < j$ (since such constraints are never tight). In addition, in constraint $z_{k,j}$
we omit the term $v_r(k) x_{r,j}$ for $r < k$ since $v_r(k) = 0$.  We get  
$$ k\cdot x_{k,j} + \sum_{r=k+1}^d x_{r,j} \leq c \quad \textrm{  constraint } z_{k,j} \quad\quad\quad \forall k \geq j : k, j \in [d]. \quad\quad \quad  $$ 

\textbf{The dual linear program:}
\begin{align*}
\text{max }   &\sum_{r=1}^d y_r \\
\text{s.t.: }  &\sum_{k=1}^d \sum_{j=1}^d z_{k,j} \leq 1&  &\text{constraint $c$}&\\
& y_i \leq  i\cdot z_{i,j} + \sum_{r=j}^{i-1} z_{r,j}&   &\text{constraint $x_{i,j}$}& &\forall i \geq j: i, j \in [d]& \\
 &z_{k,j} \geq 0& && &\forall k \geq j : k, j \in [d].&
\end{align*}	
We omit constraint $c$ by normalizing all variables by the term $\sum_k\sum_j z_{k,j}$, which
ensures that the constraint corresponding to $c$ is feasible, but modifies the goal function to
$ \frac{\sum_r y_r}{ \sum_k \sum_j z_{k,j}} $.

\textbf{Intuition for the dual variables assignment:}
A natural assignment for $y_i$ is $y_i = 1$.  However, this yields a suboptimal solution to the dual linear program
and gives a lower bound of $2$ (see Appendix~\ref{apx:suboptimal} for more details).
A slightly more sophisticated assignment is $y_i = \frac{1}{i}$. To find an assignment for variables $z_{k,j}$, we look at the constraints corresponding to
$x_{i,j}$ as if they were `continuous' (i.e., for large $i\geq j$).  In doing so, we consider a differentiable function $f_j(x)$,
where $f'_j(k)$ approximately represents the variable $z_{k,j}$.  With this interpretation, the constraint corresponding to $x_{i,j}$
can be viewed as:

$$  i \cdot f'_j(i) + \int_j^i f'_j(x) dx \geq \frac{1}{i} \iff x \cdot f_j'(x) + f_j(x) - f_j(j) \geq \frac{1}{x}.$$ 
By solving this differential equation (assuming equality) with the boundary condition $f_j(j) = 0$ (since the variables
$z_{k,j}$ only exist for $k \geq j$), we get
$$ f_j(x) = \frac{\ln(x/j)}{x}, \quad f'_j(x)=\frac{1-\ln(x/j)}{x^2}.$$ 

\textbf{Formal feasible dual variables assignment:}
$$y_i = \frac{1}{i} , \quad z_{k,j} =  \begin{cases} 
\frac{1-\ln(k/j)}{k^2}&  \text{if } j \leq k  \leq  \lfloor e\cdot j \rfloor, \\
0 & \text{otherwise}.
\end{cases}
$$
With this assignment, we need to verify that constraint $x_{i,j}$ is feasible. For $i\leq \lfloor e \cdot j \rfloor$, we get
$$ i \cdot \frac{1-\ln(i/j)}{i^2} + \sum_{r=j}^{i-1} \frac{1-\ln(r/j)}{r^2} \geq \frac{1}{i} \iff
 \sum_{r=j}^{i-1} \frac{1-\ln(r/j)}{r^2} \geq  \frac{\ln(i/j)}{i},   $$ 
which holds due to Claim~\ref{cla:f} (we apply the claim with $f(x) = \frac{\ln(x/j)}{x}$).
For $i> \lfloor e \cdot j \rfloor$, we get
$$\sum_{r=j}^{\lfloor e \cdot j \rfloor} \frac{1-\ln(r/j)}{r^2} \geq \frac{1}{e \cdot j} \geq \frac{1}{i}, $$
where the first inequality is by Claim~\ref{cla:f}.

\textbf{Evaluating the goal function:}
Recall that $H(d) = 1 + \frac{1}{2} + \cdots + \frac{1}{d}$ denotes the $d^{th}$ harmonic number.  Applying Claim~\ref{cla:f}, we have 
$\sum_{r=j}^{\lfloor e \cdot j \rfloor} z_{r,j} \leq \frac{1}{e j} + \frac{1}{j^2}$.  This yields

$$  \frac{\sumd_{i=1}^d y_i}{ \sumd_{k=1}^d \sumd_{j=1}^d z_{k,j}} \geq \frac{H(d)}{\frac{H(d)}{e} + \sumd_j \frac{1}{j^2}} \rightarrow e,$$
since $\sum_{j=1}^d \frac{1}{j^2}$ is bounded by a constant, and $H(d) \rightarrow \infty$ as $d \rightarrow \infty$. 

\end{proof}

\section{Online Ad-auctions}\label{sec:oaa}
In the $d$-bounded online Ad-auctions problem, there are $n$ bidders which are known up front, each
with a budget of $B(i)$.  Products arrive online, and for each product $j$, at most $d$ bidders
are interested in buying the product.  Each such interested buyer $i$ bids $b_{i,j}$ for product $j$.
The mechanism then allocates product $j$ to a buyer, and gains a revenue of $b_{i,j}$ (a buyer cannot be
charged more than their budget).  The objective is to maximize the total revenue.
In the fractional version of the problem, the algorithm may sell fractions of each item $j$ to
multiple buyers.

We give a randomized lower bound of $1-\left(1-\frac{1}{d}\right)^d$ for the $d$-bounded online Ad-auctions problem (which
approaches $1 - \frac{1}{e}$).  The lower bound we present in this section can also be applied
to achieve the same lower bound against randomized algorithms for the online matching problem~\cite{KVV90} where
arriving nodes have a degree of at most $d$.

\begin{theorem}
There does not exist a randomized algorithm with a competitive ratio strictly better than $1-\left(1 - \frac{1}{d}\right)^{d}$ for the
$d$-bounded online Ad-auctions problem.
\end{theorem}
\begin{proof} 
$ $ \\
\indent \textbf{Sequence definition:}
Let $n = d^{d-1}$ be the number of initial bidders.
Each bidder $i$ has a budget of $B(i) = 1$. Moreover, for every bidder $i$ and product $j$ that $i$ is interested in, we have $b_{i,j} = 1$.
Our sequence is composed of $d-1$ phases, in addition to a final phase. In each phase $k \in \{1,\ldots,d-1\}$, the adversary only sells to
some number of bidders $R_k$.  In particular, the adversary introduces $R_k/d$
products to $R_k$ bidders by grouping the $R_k$ bidders into $R_k/d$ groups (each of size $d$) and introduces a product for
each group. A bidder is in the set $[R_k]$ (recall that $[n] = \{1,\ldots,n\}$) if they were in the set $[R_{k-1}]$ and did not
have the highest leftover budget in their group from phase $k-1$ (in other words, from phase $k-1$ to $k$, the adversary drops a buyer
from each group).  We assume that the bidders
are reindexed at the beginning of each phase, so that each bidder's index is in the set $[R_k]$.
Note that $R_k = n \cdot \left(\frac{d-1}{d}\right)^{k-1}$ (since $R_1 = n$).
In the last phase, the adversary introduces a single product to each of the remaining $R_d$ players. 
Clearly, $OPT$ is $n$ since it can sell a product to each omitted bidder along with the rest of the bidders
in the final phase. \\

\textbf{The primal linear program:}
\begin{itemize}
	\item $x_{k,i}$ -- the amount sold to the $i^{th}$ player in the $k^{th}$ phase ($i \in R_k$).
	\item $t_{k,i}$ -- the total amount sold to the $i^{th}$ player up to the $k^{th}$ phase  ($i \in R_{k+1}$).
\end{itemize}
Note that the $i^{th}$ player in the $k^{th}$ iteration might not be the $i^{th}$ player in a different
iteration $k'$ (due to reindexing).
We associate two variables with each player $i$ in phase $k$: $x_{k,i}$ and $t_{k-1,i}$.

Let $\hat{i} = \lceil \frac{i}{d} \rceil$, and $G_a = \{(a-1)\cdot d+1,\ldots, a\cdot d-1\}$ for
all $a \in [R_k/d]$ (for a fixed phase $k$).  Note that the set $G_a \cup \{a \cdot d\}$ represents a group of $d$ players (namely, the players in group $a$ are those who
are interested in a particular product).  We exclude player $a \cdot d$ from the set $G_a$ for notational convenience.
Recall that the notation $a \mid b$ means that $a$ divides $b$.
When writing our linear programs, we refer to each constraint by its associated dual variable.
\begin{align*}
\text{max }   &\sumd_{k=1}^{d-1} \sumd_{i=1}^{R_k} x_{k,i} + \sumd_{i=1}^{R_d}  (1-t_{d-1,i})&&& \\
\text{s.t.: } &\sumd_{i \in G_a \cup 
	\{a \cdot d\}} x_{k,i} \leq 1&&   \text{constraint $y_{k,a}$}& &\forall k \in [d-1], a \in [R_k/d]&  \\
 &x_{k,d \cdot \hat{i}} + t_{k-1, d \cdot \hat{i}}  \leq x_{k,i} + t_{k-1,i}  &&\text{constraint $w_{k,i}$}& &\forall k \in [d-1],i \in [R_k] , d \nmid i&  \\
 &\sumd_{i\in R_{k}, d \nmid i} x_{k,i} + t_{k-1,i}  = \sumd_{i\in R_{k+1}}  t_{k,i}  &&\text{constraint $z_k$}& &\forall k \in [d-1]& \\
 &x_{k,i}, t_{k,i'} \geq 0 &&& &\forall k \in [d-1],i \in [R_k],i' \in [R_{k+1}], &
\end{align*}	
where the constraints $y_{k,a}$ capture the fact that we can sell at most one product per group (one for each group $a$ in phase $k$).
The constraints $w_{k,i}$ identify the bidder with the  highest leftover budget in their group during phase $k$.  Lastly, the
constraints $z_k$ correspond to reordering the bidders.  In fact, we even allow bidders' leftover budgets to be redistributed.
The goal function is the total amount sold from phases $1$ to $d-1$ plus the remaining leftover budgets from bidders in $R_d$.
Note that we define $t_{0,i} = 0$ for all $i$.

\textbf{The dual linear program:}
\begin{align*}
\text{min } &\sumd_{k=1}^{d-1} \sumd_{a=1}^{R_k/d} y_{k,a} + R_d  \\
\text{s.t.: }& y_{k,\hat{i}} - w_{k,i} + z_k \geq 1& &\text{ constraint $x_{k,i}$}& &\forall k\in [d-1],i \in [R_k] , d \nmid i&   \\
 &y_{k,\hat{i}} + \sum_{r \in G_{\hat{i} }}w_{k,r}  \geq 1& &\text{ constraint $x_{k,i}$}& &\forall k\in [d-1],i \in [R_k] , d \mid i&  \\
 & -w_{k+1,i} - z_{k} + z_{k+1} \geq 0    &&\text{ constraint $t_{k,i}$}& &\forall k\in [d-2],i \in [R_{k+1}] , d \nmid i&  \\
  &\sum_{r \in G_{\hat{i} }} {w_{k+1,r}} - z_{k}  \geq 0   && \text{ constraint $t_{k,i}$} &&\forall k\in [d-2],i \in [R_{k+1}] , d \mid i& \\
  &-z_{d-1}  \geq -1  &&\text{ constraint $t_{d-1,i}$}& &\forall i \in [R_{d}] \text{ (the same constraint)}& \\
  &y_{k,\hat{i}}, w_{k,i} \geq 0 &&& &\forall k \in [d-1],i \in [R_k], d \nmid i&
\end{align*}	
Note that the constraint corresponding to $t_{d-1,i}$ is independent of $i$ and appears $d-1$ times in the dual linear program.

\textbf{Intuition for the dual variables assignment:}
By symmetry, we assume that $w_{k,i} = w_{k,i'}$ for all $i,i'$ and denote this common value by $w_k$.
Similarly, we assume that $y_{k,a} = y_k$ for all $a$.  In addition, we assume that all constraints are tight $\forall k \in [d-1]$, which yields:
\begin{eqnarray}
&&y_{k} - w_{k} + z_k = 1, \label{ad:xki} \\
&&y_{k}+(d-1) w_k = 1, \label{ad:xkti} \\
&&-w_{k+1}-z_k+z_{k+1} = 0,  \label{ad:tki} \\
&&(d-1)w_{k+1} - z_k = 0, \label{ad:tkti} \\
&&z_{d-1} = 1. \label{ad:tdm}
\end{eqnarray}

\textbf{Formal feasible dual variables assignment:}
From Equations (\ref{ad:tki}) and (\ref{ad:tkti}), we derive
$ z_{k+1} = \left(\frac{d}{d-1}\right) \cdot z_k$.  Due to Equation (\ref{ad:tdm}), this yields
$$z_k = \left(\frac{d-1}{d}\right)^{d-k-1}.$$
From Equation (\ref{ad:tkti}) and Equation (\ref{ad:xkti}), we get
$$w_{k} = \frac{\left(\frac{d-1}{d}\right)^{d-k}}{d-1}, \quad y_{k} = 1-\left(\frac{d-1}{d}\right)^{d-k},$$
respectively. Finally, we need to verify that Equation (\ref{ad:xki}) holds:
\begin{eqnarray*}
y_{k} - w_{k} + z_k &=& 1-\left(\frac{d-1}{d}\right)^{d-k} - \frac{\left(\frac{d-1}{d}\right)^{d-k}}{d-1} + \left(\frac{d-1}{d}\right)^{d-k-1} \\
&=& 1 + \left(\frac{d-1}{d}\right)^{d-k-1}\left(-\frac{d-1}{d}-\frac{1}{d}+1\right) = 1.
\end{eqnarray*}
\noindent\textbf{Evaluating the goal function:}
\begin{eqnarray*}
R_d + \sumd_{k=1}^{d-1} y_{k} \cdot  \frac{R_k}{d} &=& n\cdot \left( \left(\frac{d-1}{d}\right)^{d-1} +\frac{\sumd_{k=1}^{d-1} \left(\frac{d-1}{d}\right)^{k-1}} {d}  
-(d-1)\frac{ \left(\frac{d-1}{d}\right)^{d-1}} {d}  \right)\\
&=& n \cdot \left(1-\left(\frac{d-1}{d}\right)^{d}\right).
\end{eqnarray*} 
\end{proof}

\section{Online Capital Investment}\label{sec:oci}
In this section, we study the online Capital Investment problem
We must produce many units of a commodity at minimum cost, where orders for units arrive online.
We have a set of machines, where each machine $m_i$ has a capital cost $c_i$ and production cost $p_i$.
At any time, the algorithm can choose to buy any machine for cost $c_i$.  The algorithm incurs a production
cost of $p_i$ if it uses machine $m_i$ to produce one unit of the commodity.  The goal is to minimize the total cost:
the sum of capital costs plus production costs.

We give a randomized lower bound that is arbitrarily close to $e$ for the Capital Investment problem.
In addition, our techniques enable us to give a different lower bound for each input configuration
(i.e., values for machine capital costs and production costs) that is tight
against each such configuration.
\begin{theorem}
There does not exist a randomized algorithm with a competitive ratio strictly better than $e$ for the
Capital Investment problem.
\end{theorem}

\begin{proof} $ $

\indent \textbf{Sequence definition:} 
The setting for the problem is a set of $n$ machines where machine $m_i$ has a capital cost of $i+1$ and a production cost of $2^{-i^2}$.  Our input sequence consists of $n$ phases, where in phase $k$
the online algorithm needs to produce a total of $2^{k^2}$ products.  That is, we introduce $2^{k^2} - 2^{(k-1)^2}$
orders for units in phase $k$, with two orders being introduced in the first phase.
Clearly, $OPT$ is $k+2$ in phase $k$.
	
\textbf{The primal linear program:}

\begin{itemize}
	\item $x_{k,i}$ -- the fraction bought of the $i^{th}$ machine in the $k^{th}$ phase.
	\item $q_{k,i}$ -- the fraction of products produced by the $i^{th}$ machine in the $k^{th}$ phase.
	\item $c$ -- a variable representing the competitive ratio guarantee.
\end{itemize}

\begin{align*}
\text{min } &  c && &\\
\text{s.t.: }&  \sum_{r=1}^k x_{r,i} \geq q_{k,i}  &&\text{constraint $y_{k,i}$}& &\forall k, i \in [n]&\\
 &\sum_{i=1}^{n} q_{k,i}  = 1  &&\text{constraint $w_k$}& &\forall k \in [n]& \\
  & \sumd_{r=1}^{k} \sumd_{i=1}^{n} (i+1) \cdot x_{r,i}  +2^{k^2} \sumd_{i=1}^n 2^{-i^2} q_{k,i} \leq c \cdot (k+2)  &&\text{constraint $z_k$} &&\forall k \in [n]& \\
  &c,x_{k,i}, q_{k,i} \geq 0 &&& &\forall k,i \in [n],&
\end{align*}	
where constraints $y_{k,i}$ capture the fact that a machine cannot be used more the amount paid for it.  The constraints $w_k$ imply that we need to use
a machine to produce units.  Constraints $z_k$ capture the competitive ratio guarantee of the online algorithm.  Namely, in phase $k$, the online algorithm's
total capital costs plus production costs must not exceed $c \cdot OPT$.  Note that we allow the online algorithm to produce all products in phase $k$,
which allows the online algorithm to be refunded for the production done in previous phases.  This of course does not decrease $OPT$ and may only
decrease the competitive ratio.

\textbf{The dual linear program:}
\begin{align*}
\text{max } &\sumd_{k = 1}^n w_k&&&  \\
\text{s.t.: }&  \sumd_{k = 1}^n (k+2)\cdot z_k \leq 1 &&\text{ constraint $c$} \\
  &(i+1) \sumd_{r=k}^{n} z_r \geq \sumd_{r=k}^{n} y_{r,i} && \text{ constraint $x_{k,i}$}& \forall k, i \in [n]  \\
  & y_{k,i} \geq   w_k - z_k \cdot 2^{k^2-i^2}    &&\text{ constraint $q_{k,i}$}& \forall k,i \in [n] \\
  &y_{k,i}, z_{k} \geq 0 &&& \forall k,i \in [n]&.
\end{align*}	
Note that we can omit constraint $c$ by replacing the goal function with $\frac{ \sum w_k}{\sum (k+2) z_k}$,
which is achieved by normalizing all dual variables appropriately.

\textbf{Intuition for the dual variables assignment:}
A natural assignment is $y_{k,i} = w_k$ for $k \leq i$ and $0$ otherwise (since $2^{k^2-i^2}$ should dominate $w_k/z_k$),
along with $w_k=\frac{1}{k}$.  This assignment yields
$$ \qquad (i+1) \sumd_{r=k}^{n} z_r \geq \sumd_{r=k}^{i} w_r \qquad\qquad\qquad\quad \text{ constraint $x_{k,i}$} \quad\quad\quad \forall k\geq i: k,i \in [n].$$
We view constraints $x_{k,i}$ as if they were `continuous' (i.e., for large $i\geq k$).  In doing so,
we consider two differentiable functions $f(x)$ and $g(x)$, where $f'(k)$ and $g'(i)$ represent approximately the variables
$z_{k}$ and $w_{i}$, respectively.  With this, we get $\forall i \geq k$
 $$  (i+1) \int_{k}^{n+1} f'(x) dx \geq \int_k^{i+1} g'(x) dx \iff f(n+1) - f(k)  \geq \frac{\ln(\frac{i+1}{k})}{i+1}, $$
which holds for 
$$ f(x) = -\frac{1}{e\cdot x}, \quad f'(x) = \frac{1}{e\cdot x^2},$$
since $\ln(x)/x\leq 1/e$ for $x\geq 1$.  We assume that $f(n) \rightarrow 0$ as $n \rightarrow \infty$ and ignore it.

\textbf{Formal feasible dual variables assignment:} Let $\epsilon$ be a small constant, the assignment is:
\begin{itemize}
	\item $y_{k,i} = w_k$, for $k \leq i \leq n$ and $0$ otherwise.
	\item $z_k = \frac{1}{k(k+1)} $, for all $k \leq n$.
	\item $w_k=e \cdot (1-\epsilon) \ln\left(\frac{k+1}{k}\right)$, for $ k\leq n\cdot \epsilon$ and $0$ otherwise.
\end{itemize} 		

First, we verify that constraints $q_{k,i}$ hold for $k > i \geq 1$ (they trivially hold for $k\leq i$):

$$ w_{k} - z_{k} \cdot 2^{k^2-i^2} \leq e \ln\left(\frac{k+1}{k}\right) - \frac{ 2^{2\cdot k-1}}{k(k+1)} \leq 0 = y_{k,i}.  $$

Now we verify that constraints $x_{k,i}$ hold
for all $i \leq  \epsilon \cdot n$ (since for $i >  \epsilon \cdot n$ the left hand side of constraint $x_{k,i}$
increases while the right hand side remains the same as $i$ increases). 
By assigning  to the dual variables:
\begin{eqnarray*}
(i+1) \sumd_{r=k}^{n} z_r  - \sumd_{r=k}^{i} w_r &=& (i+1) \left(\frac{1}{k}-\frac{1}{n+1}\right) - e \cdot \ln\left(\frac{i+1}{k}\right)\cdot( 1-\epsilon) \\ &=&
( 1-\epsilon)  \left(\frac{i+1}{k} - e \cdot \ln\left(\frac{i+1}{k}\right)\right) + \left(\frac{\epsilon \cdot(i+1)}{k} - \frac{i+1}{n+1}\right) \\ &\geq& 0,
\end{eqnarray*} 
where the inequality holds since $x \geq e \cdot \ln(x)$ for all $x\geq 1$ and $ k\leq i \leq \epsilon \cdot  n$.
 
\textbf{Evaluating the goal function:}
Recall that $H(n) = 1 + \frac{1}{2} + \cdots + \frac{1}{n}$ denotes the $n^{th}$ harmonic number.  We get a lower bound of
$$ \frac{\sumd_{k=1}^n w_k}{\sumd_{k=1}^n (k+2) z_k} = \frac{e \cdot \ln(n\cdot \epsilon)\cdot (1-\epsilon)}{\sum_{k=1}^n\frac{k+2}{k(k+1)}} \geq \frac{e \cdot \ln(n\cdot \epsilon)\cdot (1-\epsilon)}{H(n) + C_1} \rightarrow e \cdot (1-\epsilon),$$
where $C_1$ and $\epsilon$ are some constants.  This gives the theorem.
\end{proof}

\section{Conclusions}
We introduce a new technique for proving online lower bounds using duality.  Using our framework,
we show how to construct new, tight lower bounds for three diverse online problems.  In particular,
we give tight lower bounds for online Vector Bin Packing for arbitrarily large dimensions, online Ad-auctions in the bounded
degree setting, and online Capital Investment.  As a corollary, we also obtain tight lower bounds for online
bipartite matching in the bounded degree setting.  We are also able to reconstruct many existing online lower bounds.
We are certain that the techniques we develop here can have far-reaching implications, and can be used to improve existing
lower bounds along with proving new, tight lower bounds as well.

\newpage
\appendix
\section*{Appendix}
\section{Suboptimal Assignment for Vector Bin Packing}\label{apx:suboptimal}
\textbf{Intuition for the dual variables assignment:}
We note that by assigning $y_i = 1$ for all $i$ leads to the following non-optimal solution (with a value approaching 2).
To find an assignment for variables $z_{k,j}$, we look at the constraints corresponding to
$x_{i,j}$ as if they were `continuous' (i.e., for large $i\geq j$).  With this interpretation, we get:
$$  i \cdot f'_j(i) + \int_j^i f'_j(x) dx \geq 1 \iff x \cdot f_j'(x) + f_j(x) - f_j(j) \geq 1.$$
By solving this differential equation (assuming equality) with the boundary condition $f_j(j) = 0$, we get:
$$ f_j(x) = 1-\frac{j}{x}, \quad f'_j(x)=\frac{j}{x^2}.$$

\textbf{Formal feasible dual variables assignment:}
$$y_i = 1 , \quad z_{k,j} =  \begin{cases} 
\frac{j-1}{k\cdot (k-1)}&  \text{if } k>1, k \geq j,   \\
1 & \text{if } k = j = 1,\\
0 & \text{otherwise}.
\end{cases}
$$

With this assignment, we need to verify that constraint $x_{i,j}$ is feasible for all $i \geq j$. For $j=1$ and for $i \geq j$, the constraint corresponding to $x_{i,1}$
easily holds, since both sides of the inequality evaluate to $1$.  For $j > 1$, we get:
$$ i \cdot \frac{j-1}{i \cdot (i-1)} + \sum_{r=j}^{i-1} \frac{j-1}{r \cdot (r-1)} = 
\frac{j-1}{i-1} + (j-1) \cdot \left( \frac{1}{j-1} - \frac{1}{i-1} \right)  = 1. $$

\textbf{Evaluating the goal function:}
Note that $\sum_{k=j}^d z_{k,j} =  1 - \frac{j-1}{d},$
and hence the value of the goal function is
$$  \frac{\sumd_{i=1}^d y_i}{ \sumd_{k=1}^d \sumd_{j=1}^d z_{k,j}} = \frac{d}{d - \sumd_{j=1}^d \frac{j-1}{d}} = \frac{d}{d-\frac{d-1}{2}}  %= \frac{d}{\frac{d+1}{2}} 
 \rightarrow 2.$$

\section{Simulating Randomized Algorithms with Deterministic Fractional Algorithms}\label{apx:rvbp}
In this section, given a randomized algorithm, we show how to simulate it using a deterministic fractional algorithm.
In particular, we guarantee that the performance
of the simulated deterministic fractional algorithm for any sequence is the same as the expected performance
of the randomized algorithm.  Therefore, a lower bound for any deterministic fractional algorithm
is also a lower bound for any randomized algorithm.
In general, such a simulation is done by considering the randomized algorithm's expected behavior,
and viewing it as a deterministic fractional algorithm where the fractions are the expected probabilities
of the randomized algorithm (with the appropriate definition of a fractional algorithm).

For the online Ad-auctions problem and the online Capital Investment problem, this is straightforward.
In particular, for the online Ad-auctions problem, the expectation of the randomized algorithm's
behavior is equivalent to splitting each item.  For the online Capital Investment problem, the
randomized algorithm's behavior is equivalent to purchasing fractional machines.

For the Vector Bin Packing problem, we need to be more careful about our definition of a fractional algorithm,
since the behavior of the randomized algorithm includes both splitting vectors randomly and
opening bins randomly.  A randomized algorithm that splits vectors randomly can be viewed as a deterministic
algorithm splitting vectors (according to the expectation), while opening bins randomly can be viewed
as opening a fractional amount of bins (according to the expectation).

\bibliographystyle{plain}
\bibliography{refs}

\end{document}